\documentclass{amsart}
\usepackage[latin1]{inputenc}
\usepackage[T1]{fontenc}
\usepackage{amsmath,amsfonts,amssymb,amsthm,url,mathrsfs,color,esint,bm,mathtools,tikz,stmaryrd}
\usepackage{bbm,comment}

\newcommand{\iy}{\infty}
\newcommand{\st}{\ : \ }

\DeclareMathOperator{\tr}{Tr}

\DeclareMathOperator{\Prob}{Prob}

\renewcommand{\leq}{\leqslant}
\renewcommand{\geq}{\geqslant}

\renewcommand{\subset}{\subseteq}

\newcommand{\gC}{\mathsf{C}}

\renewcommand{\P}{\mathbf{P}}
\newcommand{\R}{\mathbf{R}}

\newcommand{\CC}{\mathbf{C}}

\newcommand{\N}{\mathbf{N}}

\DeclareMathOperator{\conv}{\mathrm{conv}}
\DeclareMathOperator{\cone}{\mathrm{cone}}

\DeclareMathOperator{\card}{\mathrm{card}}

\DeclareMathOperator{\mathspan}{\mathrm{span}}
\DeclareMathOperator{\aff}{\mathrm{aff}}

\newcommand{\ketbra}[2]{| #1 \rangle \langle  #2 |}

\newcommand{\M}{\mathsf{M}}

\newcommand{\Id}{\mathrm{Id}}

\def\beq{\begin{equation}}
\def\eeq{\end{equation}}
\def\bq{\begin{quote}}
\def\eq{\end{quote}}
\def\ben{\begin{enumerate}}
\def\een{\end{enumerate}}
\def\bit{\begin{itemize}}
\def\eit{\end{itemize}}

\def\lb{\left(}
\def\rb{\right)}
\def\lset{\lbrace}
\def\rset{\rbrace}

\def\r|{\right|}

\newcommand{\Ext}{\mathsf{Ext}}
\newcommand{\Sym}{\mathrm{Sym}}

\DeclareMathOperator{\inter}{int}
\DeclareMathOperator{\Tr}{Tr}
\DeclareMathOperator{\Av}{Av}

\newcommand{\ket}[1]{| #1 \rangle}

\theoremstyle{plain}
\newtheorem{theorem}{Theorem}
\newtheorem{corollary}[theorem]{Corollary}
\newtheorem{proposition}[theorem]{Proposition}
\newtheorem{lemma}[theorem]{Lemma}

\newtheorem{remark}{Remark}

\newcommand{\tmin}{{\otimes_{\min}}}
\newcommand{\tmax}{{\otimes_{\max}}}

\title{Monogamy of entanglement between cones}
\author{Guillaume Aubrun}
\address{\small{Institut Camille Jordan, Universit\'{e} Claude Bernard Lyon 1, 43 boulevard du 11 novembre 1918, 69622 Villeurbanne cedex, France}}
\address{\small{Univ.\ Lyon, ENS Lyon, UCBL, CNRS, Inria, LIP, F-69342, Lyon Cedex 07,
France.}}
\email{aubrun@math.univ-lyon1.fr}

\author{Alexander M\"uller-Hermes}
\address{\small{Department of Mathematics, University of Oslo, P.O. box 1053, Blindern, 0316 Oslo, Norway}}
\email{muellerh@math.uio.no}

\author{Martin Pl\'avala}
\address{\small{Naturwissenschaftlich-Technische Fakult\"{a}t, Universit\"{a}t Siegen, 57068 Siegen, Germany}}
\email{martin.plavala@uni-siegen.de}


\begin{document}

\maketitle

\begin{abstract}
    A separable quantum state shared between parties $A$ and $B$ can be symmetrically extended to a quantum state shared between party $A$ and parties $B_1,\ldots ,B_k$ for every $k\in\N$. Quantum states that are not separable, i.e., entangled, do not have this property. This phenomenon is known as ``monogamy of entanglement''. We show that monogamy is not only a feature of quantum theory, but that it characterizes the minimal tensor product of general pairs of convex cones $\gC_A$ and $\gC_B$: The elements of the minimal tensor product $\gC_A\otimes_{\min} \gC_B$ are precisely the tensors that can be symmetrically extended to elements in the maximal tensor product $\gC_A\otimes_{\max} \gC^{\otimes_{\max} k}_B$ for every $k\in\N$. Equivalently, the minimal tensor product of two cones is the intersection of the nested sets of $k$-extendible tensors. It is a natural question when the minimal tensor product $\gC_A\otimes_{\min} \gC_B$ coincides with the set of $k$-extendible tensors for some finite $k$. We show that this is universally the case for every cone $\gC_A$ if and only if $\gC_B$ is a polyhedral cone with a base given by a product of simplices. Our proof makes use of a new characterization of products of simplices up to affine equivalence that we believe is of independent interest.      
\end{abstract}

\section{Introduction}

Let $V$ denote a finite-dimensional vector space over the reals and $\gC\subset V$ a convex cone. We say that the cone $\gC$ is \emph{proper} if it is closed, does not contain any line, and is not contained in any hyperplane. Any proper cone $\gC\subset V$ can be described in terms of its dual $\gC^*\subset V^*$, given by 
\[
\gC^* = \lset f\in V^*~:~f(x)\geq 0 \text{ for any }x\in \gC\rset ,
\]
since by the bipolar theorem we have $(\gC^*)^*=\gC$. While any proper cone can be described from the inside by its extremal rays, the bipolar theorem describes the cone from the outside as an intersection of half-spaces. This duality is fundamental in the study of convex cones and it is a basic problem in convex analysis to switch between these two descriptions.

We will be interested in the two canonical tensor products of any pair of proper cones. Given proper cones $\gC_A\subset V_A$, $\gC_B\subset V_B$, we may form the \emph{minimal tensor product}
\[
\gC_A \tmin \gC_B = \conv\lset x\otimes y~:~x\in \gC_A,\ y\in \gC_B\rset\subset V_A\otimes V_B ,
\]
and the \emph{maximal tensor product} 
\[
\gC_A \tmax \gC_B = (\gC^*_A\tmin \gC^*_B)^* = \lset z\in V_A\otimes V_B~:~ (f\otimes g)(z)\geq 0,\ f\in \gC_A^*,\ g\in \gC_B^* \rset .
\]
For proper cones $\gC_A$ and $\gC_B$ both tensor products are again proper cones. Moreover, the tensor products may be iterated and since they are associative we may define the $k$-fold versions $\gC^{\otimes_{\min} k}$ and $\gC^{\otimes_{\max} k}$ for any proper cone $\gC$. The two tensor products reflect the aforementioned descriptions of convex cones: The minimal tensor product is described by specifying generating elements and the maximal tensor product by an intersection of half-spaces. Moreover, their discrepancy can be seen as a general form of quantum entanglement \cite{ALPP21}.

A central feature of quantum physics is known as \emph{monogamy of entanglement}. In our language, it can be described using the cones of positive semidefinite $d\times d$ matrices $\M^+_d$ with complex entries, which is a proper cone inside the real vector space~$\M_d^{\mathrm{sa}}$ of Hermitian $d\times d$ matrices. A bipartite quantum state is given by a matrix $\rho_{AB}\in \M_{d_Ad_B}^+ \simeq (\M_{d_A}\otimes \M_{d_B})^+ \subsetneq \M_{d_A}^+ \tmax \M_{d_B}^+$ for some integers ${d_A,d_B \geq 2}$ and it is called \emph{entangled} if it does not belong to $\M^+_{d_A}\otimes_{\min} \M^+_{d_B}$. The \emph{monogamy theorem} \cite{DPS04, Yang06} asserts that a quantum state $\rho_{AB}\in (\M_{d_A}\otimes \M_{d_B})^+$ is entangled if and only if there exists an integer $k \geq 2$ for which no quantum state $\sigma_{AB_1\cdots B_k}\in (\M_{d_A}\otimes \M^{\otimes k}_{d_B})^+$ satisfies the equation
\[
\rho_{AB} = \lb\Id_{\displaystyle \M^{\mathrm{sa}}_{d_A}}\otimes \gamma^{\Tr}_k\rb\lb \sigma_{AB_1\cdots B_k}\rb ,
\]
where
\[
\gamma^{\Tr}_k = \frac{1}{k} \sum_{j=1}^k \Tr^{\otimes (j-1)} \otimes \Id_{\displaystyle \M^{\mathrm{sa}}_{d_B}} \otimes \Tr^{\otimes (k-j)}
\]
corresponds to (symmetrically) discarding all but one out of the $k$ tensor factors labelled by $B_1,\ldots, B_k$. Physically, this means that a bipartite quantum state is entangled if and only if it cannot be partially shared with arbitrarily many parties. Our first main result establishes a similar property for general minimal and maximal tensor products.

\begin{theorem}\label{thm:Main1}
Consider proper cones $\gC_A\subset V_A$ and $\gC_B\subset V_B$ and a linear form $\phi$ in the interior of $\gC_B^*$. Then, we have
\begin{equation} \label{eq:main-equation}
\gC_A \tmin \gC_B = \bigcap_{ k \geq 1}  \left( \Id_{V_A} \otimes \gamma_k^\phi \right) \left(\gC_A \tmax \gC_B^{\tmax k} \right),
\end{equation}
where, for an integer $k \geq 1$, the \emph{$k$th reduction map} $\gamma_k^\phi : V_B^{\otimes k} \to V_B$ is defined as
\[ \gamma_k^\phi = \frac{1}{k} \sum_{j=1}^k  \phi^{\otimes (j-1)} \otimes \Id_{V_B} \otimes \phi^{\otimes (k-j)}. \]
\end{theorem}
Our Theorem \ref{thm:Main1} shows that the monogamy property of quantum entanglement is actually a common feature of the most general forms of entanglement for any pair of cones. The question whether Theorem \ref{thm:Main1} holds has been asked in \cite[\S 2.4.1]{lamiatesi}. Even for cones of positive semidefinite matrices our results are new: They show that entanglement cannot be partially shared with arbitrarily many parties even if we allow for unphysical quantum states represented in the maximal tensor product. 

Equation \eqref{eq:main-equation} produces a decreasing sequence of outer approximations to the minimal tensor product, called the \emph{extendability hierarchy}. A natural question is whether this sequence stops after finitely many steps. We answer this question in terms of combinatorial properties of the base $K_\phi = \gC_B \cap \phi^{-1}(1)$. We have the following theorem:

\begin{theorem}\label{thm:Main2}
Consider a proper cone $\gC_B\subset V_B$ and a linear form $\phi$ in the interior of $\gC_B^*$. Then, for every integer $k \geq 1$, the following are equivalent:
\begin{enumerate}
    \item[(i)] For every proper cone $\gC_A\subset V_A$, we have the identity
    \[ \gC_A \tmin \gC_B =  \left( \Id_{V_A} \otimes \gamma_k^\phi \right) \left(\gC_A \tmax \gC_B^{\tmax k} \right) ;\]
    \item[(ii)] The base $K_\phi$ is affinely equivalent to a Cartesian product of $\leq k$ simplices.
\end{enumerate}
\end{theorem}

As a byproduct of our proof, we obtain a characterization of products of simplices as the only polytopes (up to affine equivalence) for which the operations ``intersection'' and ``affine hull'' commute when applied to the face lattice (see Figure \ref{figure-th3} for an illustration). We denote by $\aff(X)$ the affine hull of a subset $X$ of a vector space, using the convention that $\aff(\emptyset)=\emptyset$.
  
\begin{theorem} \label{theorem:polysimplices}
Let $P$ be a polytope. The following are equivalent:
\begin{enumerate}
    \item[(i)] The polytope $P$ is affinely equivalent to a Cartesian product of simplices;
    \item[(ii)] For every faces $(F_i)_{i \in I}$ of $P$, we have
    \begin{equation} \label{eq:aff} 
    \aff \left( \bigcap_{i \in I} F_i \right) = \bigcap_{i \in I} \aff ( F_i). \end{equation}
\end{enumerate}
\end{theorem}

In the case of dimension $3$, Theorem \ref{theorem:polysimplices} means that the only polyhedra without stellation are the simplex, the cube and the triangular prism, extending an observation by Wenninger \cite[p.35]{Wenninger74}.

\begin{figure}[htbp] \begin{center}
\begin{tikzpicture}[scale=1]
	\coordinate (a) at (0,0);
	\coordinate (b) at (0,1);
	\coordinate (c) at (2,2);
	\coordinate (d) at (2,0);
	\coordinate (x) at (-2,0);
	\draw[fill=gray!30] (a)--(b)--(c)--(d)--(a);
    \draw[dashed] (b)--(x)--(a);
	\end{tikzpicture}
\end{center} 
\caption{The planar case of Theorem \ref{theorem:polysimplices}: if a polygon is neither a triangle nor a parallelogram, it has disjoint edges whose affine hulls intersect.}
\label{figure-th3}
\end{figure}

\subsection*{Conventions and preliminaries} Given an integer $n \geq 1$, we denote by $[n]$ the set $\{1,\dots,n\}$. All the vector spaces we consider are implicitly assumed to be finite-dimensional. If $V$ is a vector space, the action of $f \in V^*$ on $x \in V$ is denoted as $f(x)$, $\langle f,x \rangle$ or $\langle x,f \rangle$. We always identify the double dual~$V^{**}$ with $V$ itself. Similarly, for an integer $k \geq 2$ we always identify $(V^{\otimes k})^*$ with $(V^*)^{\otimes k}$.

Given vector spaces $V_1,V_2$, the adjoint of a linear map $\Phi:V_1 \to V_2$ is denoted as $\Phi^* : V_2^* \to V_1^*$. If $X \subset V$, the affine hull $\aff(X)$ is the set of elements of the form $\lambda_1 x_1 + \dots + \lambda_n x_n$ for an integer $n \geq 1$, elements $x_1,\dots,x_n$ in $X$ and real numbers $\lambda_1,\dots,\lambda_n$ such that $\lambda_1 + \dots +\lambda_n =1$. An element $y \in V$ belongs to $\aff(X)$ if and only if every affine function $f:V \to \R$ which vanishes on $X$ vanishes at $y$.

\subsection*{Structure of the paper} After reviewing the basic properties of the symmetric subspace, we discuss the basic properties of the generalized extendibility hierarchy (Section \ref{sec:BasicsHierarchy}), a useful dual formulation (Section \ref{sec:dual}) and some examples (Section~\ref{sec:Examp}) including the case of quantum theory. In Section \ref{sec:proof-convergence} we prove Theorem~\ref{thm:Main1} by making use of a generalization of the de Finetti theorem due to Barrett and Leifer \cite{BarrettLeifer}. The proof of Theorem \ref{thm:Main2} is presented in Section \ref{sec:HierarchyFinite}. It relies on Theorem~\ref{theorem:polysimplices} which we prove in Section \ref{sec:PolySimpl}.

\section{Basic properties of the extendibility hierarchy and examples} \label{sec:results}
Let $V$ be a finite-dimensional real vector space and $k \geq 1$ an integer. The symmetric group $\mathfrak{S}_k$ acts naturally on $V^{\otimes k}$: if $\sigma \in \mathfrak{S}_k$, we define $U_\sigma : V^{\otimes k} \to V^{\otimes k}$ as the linear map such that
\[ U_\sigma (x_1 \otimes \dots \otimes x_k) = x_{\sigma^{-1}(1)} \otimes \dots \otimes x_{\sigma^{-1}(k)}\]
for every $x_1,\dots,x_k \in V$.
We define the \emph{symmetric subspace} $\Sym_k(V)$ as the subspace of $V^{\otimes k}$ which is invariant for this action
\[ \Sym_k(V) = \{ x \in V^{\otimes k} \st U_\sigma x = x \textnormal{ for every } \sigma \in \mathfrak{S}_k \}. \]
The \emph{symmetric projection} is the operator $P_{\Sym_k(V)} : V^{\otimes k} \to V^{\otimes k}$ defined as
\[ P_{\Sym_k(V)} = \frac{1}{k!} \sum_{\sigma \in \mathfrak{S}_k} U_\sigma .\]
It is a projection whose range equals $\Sym_k(V)$. If we identify $(V^{\otimes k})^*$ with $(V^*)^{\otimes k}$, we have the relation
\[ P_{\Sym_k(V)}^* = P_{\Sym_k(V^*)} .\]
Moreover, if $\gC \subset V$ is a proper cone, then
\[ P_{\Sym_k(V)}( \gC ^{\tmin k}) = \gC ^{\tmin k} \cap {\Sym_k(V)}, \]
\[ P_{\Sym_k(V)}( \gC ^{\tmax k}) = \gC ^{\tmax k} \cap {\Sym_k(V)}. \]

\subsection{The extendibility hierarchy}\label{sec:BasicsHierarchy}

Consider proper cones $\gC_A \subset V_A$, $\gC_B \subset V_B$ and $\phi \in \inter(\gC_B^*)$.
For an integer $k \geq 1$, denote by
\begin{equation} \label{def-Extk} \Ext_k(\gC_A,\gC_B,\phi) = (\Id_{V_A} \otimes \gamma_k^\phi) \left(\gC_A \tmax \gC_B^{\tmax k} \right) \end{equation}
the cone appearing in Theorem \ref{thm:Main1}. Since $\gamma_k^\phi = (\Id_{V_B} \otimes \phi^{\otimes(k-1)}) \circ P_{\Sym_k(V_B)}$, an element $x \in V_A \otimes V_B$ belongs to $\Ext_k(\gC_A,\gC_B,\phi)$ if and only if 
\[ \exists y_k \in \left( \Id_{V_A} \otimes P_{\Sym_k(V_B)} \right) \left( \gC_A \tmax \gC_B^{\tmax k} \right) \st
 x =  \left( \Id_{V_A} \otimes \Id_{V_B} \otimes \phi^{\otimes(k-1)} \right)(y_k). \] 
Such an element $y_k$ is called a \emph{$k$-extension} of $x$ and $\Ext_k(\gC_A,\gC_B,\phi)$ is called the cone of \emph{$k$-extendible elements}. Observe that
\[ \left( \Id_{V_A} \otimes P_{\Sym_k(V_B)} \right) \left( \gC_A \tmax \gC_B^{\tmax k} \right) 
= \left( V_A \otimes \Sym_k(V_B) \right) \cap \left( \gC_A \tmax \gC_B^{\tmax k} \right)
\]
and therefore, whenever $y_k$ is a $k$-extension of $x$, then $(\Id_{V_A} \otimes \Id_{V_B}^{\otimes (k-1)} \otimes \phi ) (y_k)$ is a $(k-1)$-extension of $x$. It follows that the sequence $\Ext_k(\gC_A,\gC_B,\phi)$ is a decreasing sequence of cones.
We also define 
\[
\Ext_\infty(\gC_A,\gC_B,\phi) = \bigcap_{k \geq 1}\Ext_k(\gC_A,\gC_B,\phi).
\]
The first level of this hierarchy is the maximal tensor product itself, i.e., we have $\Ext_1(\gC_A,\gC_B,\phi)=\gC_A \tmax \gC_B$, and Theorem \ref{thm:Main1} asserts that 
$\Ext_{\iy}(\gC_A,\gC_B,\phi)=\gC_A \tmin \gC_B$. 

\subsection{The dual of the extendibility hierarchy}\label{sec:dual}

The adjoint of the $k$th reduction map $\gamma_k^\phi$ is the map $(\gamma_k^\phi)^* : V_B^* \to (V_B^*)^{\otimes k}$ which for $\psi \in V_B^*$ reads as
\[ (\gamma_k^\phi)^*(\psi) = P_{\Sym_k(V_B^*)} ( \psi \otimes \phi^{\otimes(k-1)}). \]
Let $\zeta \in V_A^* \otimes V_B^*$. It follows immediately from \eqref{def-Extk} that
\begin{equation} \label{eq:Extk-dual}
\zeta \in \Ext_k(\gC_A,\gC_B,\phi)^* \iff (\Id_{V_A} \otimes (\gamma_k^\phi)^*)(\zeta) \in \gC_A^* \tmin (\gC_B^*)^{\tmin k}.
\end{equation}

In particular, we can write down a statement which is dual to Theorem \ref{thm:Main1}.

\begin{theorem} \label{theorem:dual}
Let $\gC_A \subset V_A$ and $\gC_B \subset V_B$ be proper cones. For every $x \in \inter (\gC_A \tmax \gC_B)$ and $y \in \inter (\gC_B)$, there is an integer $k \geq 1$ such that 
\[ \left( \Id_{V_A} \otimes P_{\Sym_k(V_B)} \right) (x \otimes y^{\otimes (k-1)}) \in \gC_A \tmin \gC_B^{\tmin k} .\]
\end{theorem}

\begin{proof} 
By Theorem \ref{thm:Main1} applied to $\gC_A^*$ and $\gC_B^*$, we have
\[ \gC_A^* \tmin \gC_B^* = \Ext_{\infty} ( \gC_A^*,\gC_B^*,y) = \bigcap_{k \in \N}
\Ext_k ( \gC_A^*,\gC_B^*,y). \]
Using the bipolar theorem in the form $(\gC^*)^{*}=\overline{\gC}$ whenever $\gC$ is a not-necessarily-closed convex cone, this is equivalent to the dual equation
\[ \gC_A \tmax \gC_B = \left( \bigcap_{k \in \N}
\Ext_k ( \gC_A^*,\gC_B^*,y) \right)^* = \overline{\bigcup_{k \in \N} \Ext_k ( \gC_A^*,\gC_B^*,y)^*}. \]
Since $x \in \inter (\gC_A \tmax \gC_B)$, there is an integer $k \geq 1$ such that $x \in \Ext_k ( \gC_A^*,\gC_B^*,y)^*$. By \eqref{eq:Extk-dual}, this means that $(\Id_{V_A} \otimes (\gamma_k^y)^*)(x) \in \gC_A \tmin \gC_B^{\tmin k}$, which is the desired conclusion.
\end{proof}

\subsection{Examples}\label{sec:Examp}

\subsubsection*{Simplicial cones}

We say that a cone $\gC \subset V$ is \emph{simplicial} (or \emph{classical}) if any of its bases is a simplex. Since $\gamma_1^\phi$ is the identity operator for any nonzero linear form $\phi$, the case $k=1$ of Theorem \ref{thm:Main2} says that a cone $\gC_B$ is simplicial if and only if $\gC_A \tmin \gC_B = \gC_A \tmax \gC_B$ for every proper cone $\gC_A$. This result can be traced back to \cite{NamiokaPhelps69} and is much simpler to prove than the case $k>1$ of Theorem~\ref{thm:Main2}. A strengthening of the $k=1$ case was recently obtained \cite{ALPP21}: given two proper cones $\gC_A$, $\gC_B$, we have $\gC_A \tmin \gC_B = \gC_A \tmax \gC_B$ if and only if either $\gC_A$ or $\gC_B$ is simplicial.

\subsubsection*{Cone over a square}

Consider the following vectors in $\R^3$
\begin{equation} \label{eq:square} x_{+0} = (1,1,0),\  x_{-0} = (1,-1,0),\  x_{0+} = (1,0,1),\ x_{0-} = (1,0,-1)\end{equation}
and let $\gC_B$ be the cone they generate (see Figure \ref{ee}). 

\begin{figure}[htbp] \begin{center}
\begin{tikzpicture}[scale=1]
\coordinate (O) at (0,0);
\coordinate (a) at (-2,2.5);
\coordinate (b) at (-0.4,2.1);
\coordinate (c) at (2,2.5);
\coordinate (d) at (0.4,2.9);
\coordinate (e) at (-1.2,2.3);
\coordinate (f) at (0.8,2.3);
\coordinate (g) at (1.2,2.7);
\coordinate (h) at (-0.8,2.7);

\draw (O) node {$\bullet$} ;
\draw (O) node [below] {$0$} ;

\draw (a) node [above] {$x_{-0}$} ;
\draw (c) node [above] {$x_{+0}$} ;
\draw (b) node [above] {$x_{0-}$} ;
\draw (d) node [above] {$x_{0+}$} ;

\draw (a) -- (b) -- (c) -- (d) -- (a) ; 
\draw (a) -- (O) -- (b) ;
\draw (c) -- (O) ;
\draw[gray!50, dashed] (d) -- (O) ;

\end{tikzpicture}
\end{center}
\caption{The cone $\gC_B$ over a square}
\label{ee}
\end{figure}

The dual cone $\gC^*_B$ is generated by the following linear forms, identified with vectors in $\R^3$
\[ \psi_{++} = \frac{1}{2}(1,1,1),\  \psi_{+-} = \frac{1}{2}(1,1,-1),\  \psi_{-+} = \frac{1}{2}(1,-1,1),\ \psi_{--} = \frac{1}{2}(1,-1,-1).\]

Consider the linear form $\phi=(1,0,0)$. The corresponding base $K_\phi$ of $\gC_B$ is the square whose vertices are the vectors \eqref{eq:square}.
The second reduction map satisfies  $\gamma^\phi_2(u \otimes v) = \frac{u+v}{2}$ for every $u,v \in K_\phi$. One can check that it can be rewritten, for $y \in \R^3 \otimes \R^3$, as
\begin{equation*} 
\begin{split}
2\gamma^\phi_2(y) = \ 
  &  (\psi_{++}\otimes\psi_{+-} + \psi_{+-}\otimes\psi_{++})(y) \cdot x_{+0} \\
  & +  (\psi_{-+}\otimes\psi_{--} + \psi_{--}\otimes\psi_{-+})(y) \cdot x_{-0}\\
  & +  (\psi_{++}\otimes\psi_{-+} + \psi_{-+}\otimes\psi_{++})(y) \cdot x_{0+}\\
  & +  (\psi_{+-}\otimes\psi_{--} + \psi_{--}\otimes\psi_{+-})(y) \cdot x_{0-}.
\end{split}
\end{equation*}
Let $\gC_A \subset V_A$ be a arbitrary cone. For any $z \in \gC_A \tmax \gC_B \tmax \gC_B$, this formula shows that $(\Id_A \otimes \gamma_2^\phi)(z)$ belongs to $\gC_A \tmin \gC_B$, since it produces a decomposition into a sum of $4$ terms (with nonnegative weights) of the form $x_A \otimes x_B$ with $x_A \in \gC_A$ and $x_B$ one of the vectors in \eqref{eq:square}.

We also observe that the choice of the linear form $\phi$ is crucial. If we choose $\phi' \in \inter(\gC_B^*)$ which is not proportional to $\phi$, then the corresponding base $K_{\phi'}$ is a quadrilateral which is not a parallelogram. Theorem \ref{thm:Main2} then implies that the extendibility hierarchy for $\phi'$ does not stop after finitely many steps.

\subsubsection*{The quantum case}
Consider an operator $\rho_{AB} \in (\M_{d_1} \otimes \M_{d_2})^+$ and an integer $k \geq 1$. We say that $\rho_{AB}$ is $k$-max-extendible if there exists an operator $\sigma_{AB_1\cdots B_k} \in \M_{d_1}^+ \tmax \left(\M_{d_2}^+\right)^{\tmax k}$ such that \[
\rho_{AB} = \lb\Id_{\displaystyle \M^{\mathrm{sa}}_{d_1}}\otimes \gamma^{\Tr}_k\rb\lb \sigma_{AB_1\cdots B_k}\rb .
\]
If moreover the operator $\sigma_{AB_1\cdots B_k}$ can be chosen to be positive semidefinite, we say that $\rho_{AB}$ is $k$-PSD-extendible. An operator which is $k$-PSD-extendible is also $k$-max-extendible, and the following proposition shows that the converse is false.

\begin{proposition}[Proved in Appendix \ref{app:quantum}] \label{prop:quantum}
There is an operator $\rho_{AB} \in (\M_3 \otimes \M_3)^+$ which is $2$-max-extendible but not $2$-PSD-extendible. 
\end{proposition}

The usual monogamy theorem~\cite{DPS04, Yang06} states that an operator which is $k$-PSD-extendible for every integer $k \geq 1$ is separable; the stronger version given by Theorem \ref{thm:Main1} is that an operator which is $k$-max-extendible for every $k \geq 1$ is separable. 

Finally, let us mention the following consequence of Theorem \ref{theorem:dual}. To our knowledge, this result does not appear in the quantum information literature.

\begin{corollary}
Let $\rho_{AB}\in (\M_{d_1} \otimes \M_{d_2})^+$ be of full rank. Then there exists an integer $k \geq 1$ such that the operator on $\M_{d_1} \otimes \M_{d_2}^{\otimes k}$ defined as
\[ \frac{1}{k} \sum_{i=1}^k \rho_{AB_i} = \left( \Id_{\displaystyle \M_{d_1}^{\mathrm{sa}}} \otimes P_{\Sym_k(\M_{d_2}^{\mathrm{sa}})} \right) \left(\rho_{AB} \otimes \mathbf{1}^{\otimes (k-1)}_{d_2} \right)\]
is fully separable, i.e., belongs to $\M_{d_1}^+ \tmin \left( \M_{d_2}^+\right)^{\tmin k}$. 
\end{corollary}

\section{Proof of Theorem \ref{thm:Main1}} \label{sec:proof-convergence}

One direction is easy: if $x \in \gC_A \tmin \gC_B$, we may write $x = \sum_{i} a_i \otimes b_i$ with $a_i \in \gC_A$ and $b_i \in \gC_B$. We may assume by rescaling $a_i$ that 
$\phi(b_i)=1$. The formula $y_k = \sum_i a_i \otimes b_i^{\otimes k}$ shows that $x$ is $k$-extendible for every integer $k$.

The other direction of the proof relies on a specific choice of extensions. We say that $(y_k)_{k \geq 0}$ is a \emph{compatible sequence} (for $(\gC_A,\gC_B,\phi)$) if the following conditions hold:
\begin{enumerate}
\item[(i)] We have $y_0 \in \gC_A \setminus \{0\}$;
\item[(ii)] For every $k \geq 1$, we have $y_k \in \Big( V_A \otimes \Sym_k(V_B) \Big) \cap \lb \gC_A \tmax \gC_B^{\tmax k} \rb$;
\item[(iii)] For every $k \geq 1$, we have 
\[y_{k-1} = \left( \Id_{V_A} \otimes \Id_{V_B}^{\otimes (k-1)} \otimes \phi \right)(y_k). \]
\end{enumerate}
Iterating the last property shows that whenever $(y_k)_{k \geq 0}$ is a compatible sequence, then $y_k$ is a $k$-extension of $y_1$ for every $k \geq 1$.

Our proof strategy is to first show in Lemma~\ref{lemma:compatible-extension} that for every $x \in \Ext_{\iy}(\gC_A,\gC_B,\phi)$ the extensions can be chosen such that they form a compatible sequence. Then we use a representation of compatible sequences given by Proposition~\ref{proposition:main-compatible-extensions} to finish the proof.

\begin{lemma} \label{lemma:compatible-extension}
Let $x \in \Ext_{\iy}(\gC_A,\gC_B,\phi)$. Then there exists a compatible sequence $(y_k)_{k \geq 0}$ such that $y_1=x$.
\end{lemma}

\begin{proof}
We necessarily have $y_1=x$ and $y_0 = (\Id_{V_A} \otimes \phi)(x)$. For every $k \geq 2$, let $x_k$ be an arbitrary $k$-extension of $x$. The compatibility can be enforced by a compactness argument. For $k \leq n$, the vector 
\[ y_{k,n} = ( \Id_{V_A} \otimes \Id_{V_B^{\otimes k}} \otimes  \phi^{\otimes (n-k)} )(x_n) \]
is a $k$-extension of $x$. Since for every integer $k$ the set of $k$-extensions of $x$ is compact, by a diagonal extraction process (see, e.g., \cite[Theorem I.24]{ReedSimon72}), we may find an increasing function $g : \N \to \N$ such that the limit $y_k = \lim_{n \to \infty} y_{k,g(n)}$ exists for every $k \geq 2$. The sequence $(y_k)_{k \geq 0}$ is compatible, as needed.
\end{proof}


\begin{proposition} \label{proposition:main-compatible-extensions}
Let $\gC_A \subset V_A$, $\gC_B \subset V_B$ be proper cones and $\phi \in \inter( \gC_B^*)$. Let $(y_k)_{k \geq 0}$ be a compatible sequence. Then there exists a pair $(\pi,\alpha)$, where
\begin{enumerate}
    \item[(i)] $\pi$ is a Borel probability measure on $K_\phi$,
    \item[(ii)] $\alpha : K_\phi \to V_A$ is a Borel map such that $\pi(\{ \alpha \in \gC_A\})=1$,
    \item[(iii)] for every integer $k \geq 0$,
\begin{equation}
\label{eq:main-compatible-extensions}
y_k = \int_{K_\phi} \alpha(\omega) \otimes \omega^{\otimes k} \, \mathrm{d} \pi(\omega). \end{equation}
\end{enumerate}
Moreover, if $(\pi',\alpha')$ is another pair satisfying these three conditions, then $\pi=\pi'$ and $\pi(\{\alpha=\alpha'\})=1$.
\end{proposition}



The result of Proposition~\ref{proposition:main-compatible-extensions} is very similar to (and could be deduced from) a theorem by Barrett and Leifer \cite{BarrettLeifer} which is a version of the de Finetti theorem for arbitrary cones (see also \cite{ChristandlToner09} for related results). This extends the argument given in \cite{CFS02} for the quantum case. For the reader's convenience, we will present a new simpler proof of Proposition~\ref{proposition:main-compatible-extensions} using directly the classical de Finetti theorem.

Let us recall the statement of the de Finetti theorem. Let $F$ be a finite set and $(Y_n)_{n \geq 1}$ be a sequence of $F$-valued random variables which is \emph{exchangeable} (i.e., for any integer $n \geq 1$ and any permutation $\sigma \in \mathfrak{S}_n$, the tuples $(Y_1,\dots,Y_n)$ and $(Y_{\sigma(1)},\dots,Y_{\sigma(n)})$ have the same distribution). The de Finetti theorem asserts that there is a unique Borel probability measure~$\pi$ on the set $\Prob(F)$ of probability measures on $F$ such that for every $n \geq 1$ and $x_1,\dots,x_n \in F^n$,
\[ \P(X_1=x_1,\dots,X_n=x_n) = \int_{\Prob(F)}  p(x_1)\dots p(x_n) \, \mathrm{d}\pi(p). 
\]
We need a small extension of this statement. Assume that $W$ is an additional random variable taking values in a finite set $E$ and such that the sequence $(X_n)$ is exchangeable conditionally on $W$. Then there is a unique family $(\pi_i)_{i \in E}$ of positive Borel measures on $\Prob(F)$ such that, for $i \in E$
\begin{equation} \label{eq:conditional-de-Finetti} \P(W=i,X_1=x_1,\dots,X_n=x_n)  = \int_{\Prob(F)}  p(x_1)\dots p(x_n) \, \mathrm{d}\pi_i(p). 
\end{equation}
We necessarily have $\pi = \sum_{i \in E} \pi_i$. For $i \in E$, denote by $\alpha_i$ the Radon--Nikodym derivative of $\pi_i$ with respect to $\pi$. We can rewrite \eqref{eq:conditional-de-Finetti} as
\begin{equation} \label{eq:conditional-de-Finetti2} \P(W=i,X_1=x_1,\dots,X_n=x_n)  = \int_{\Prob(F)} \alpha_i(p) p(x_1)\dots p(x_n) \, \mathrm{d}\pi(p). 
\end{equation}

We also use the following simple lemma.

\begin{lemma} \label{lemma:intersection-of-simplices}
Let $K \subset \R^n$ be a convex body, i.e., a compact convex set with non-empty interior. There is a sequence $(\Delta_k)_{k \geq 1}$ of simplices such that $\bigcap \Delta_k = K$.
\end{lemma}

\begin{proof}
Without loss of generality, assume that $0 \in K$. It suffices to show that for every $x \in \R^n \setminus K$, there is a simplex $\Delta$ containing $K$ but not $x$. This can be seen as follows: by the Hahn--Banach separation theorem, there is a linear form $f_1$ such that $f_1(x) > B := \max_K f_1$. We may find linear forms $f_2,\dots,f_{n+1}$ such that $0 \in \inter \conv (f_1,\dots,f_{n+1})$. For $A>0$, the set
\[ \Large\{ x \in \R^n \st f_1(x) \leq B,\  f_i(x) \leq A \textnormal{ for } i \in \{2,\dots,n+1\} \Large\} \]
is a simplex which does contain $x$; for $A$ large enough, it contains $K$ as needed.
\end{proof}

\begin{proof}[Proof of Proposition \ref{proposition:main-compatible-extensions}]
Suppose first that the cones $\gC_A$ and $\gC_B$ are simplicial. Without loss of generality, we may then identify $\gC_A$ and $\gC_B$ with the cone of positive measures on some finite sets~$E$ and $F$, and $K_\phi$ with the set of probability measures on $F$. Up to rescaling, we may also assume that~$y_0$ is a probability measure on $E$. In that case, the conclusion is equivalent to the conditional de Finetti theorem as in \eqref{eq:conditional-de-Finetti2}.

Consider now the general case. We use repeatedly the fact that the sequence $(y_k)_{k \geq 0}$ is a compatible sequence for $(\gC,\gC',\phi)$ whenever $\gC$ and $\gC'$ are proper cones containing respectively $\gC_A$ and~$\gC_B$, and such that $\phi \in \inter(\gC'^*)$. In particular, choosing $\gC$ and $\gC'$ to be simplicial cones, we obtain the existence of
a positive Borel measure $\pi=\pi_{\gC,\gC'}$ on $\Delta := \gC' \cap \phi^{-1}(1)$ and a Borel map $\alpha=\alpha_{\gC,\gC'}$ such that for every $k \geq 0$
\[ y_k = \int_{\Delta} \alpha(\omega) \otimes \omega^{\otimes k} \, \mathrm{d} \pi(\omega). \]
To finish the proof, it suffices to show that (i) the measure $\pi$ is supported on $K_\phi$ and (ii) the function $\alpha$ is $\pi$-almost surely $\gC_A$-valued; the uniqueness property for $(\gC_A,\gC_B)$ then follows from the uniqueness property for $(\gC,\gC')$.

Let $\Delta_1$ be another simplex in the affine hyperplane $\phi^{-1}(1)$ such that $K_\phi \subset \Delta_1$ and $\Delta_2$ a third simplex such that $\Delta \cup \Delta_1 \subset \Delta_2$. By the uniqueness property for $(\gC,\cone(\Delta_2))$, it follows that $\pi=\pi_{\gC,\cone(\Delta_1)}=\pi_{\gC,\cone(\Delta_2)}$. In particular, the measure $\pi$ is supported in $\Delta \cap \Delta_1$. Applying Lemma \ref{lemma:intersection-of-simplices} proves that $\pi$ is supported on $K_\phi$.

Similarly, let $\gC_1$ be another simplicial cone such that $\gC_A\subset \gC_1$ and $\gC_2$ a third simplicial cone such that $\gC \cup \gC_1\subset \gC_2$. By the uniqueness property for $(\gC_2,\gC')$,  we have $\alpha=\alpha_{\gC_1,\gC'}=\alpha_{\gC_2,\gC'}$ $\pi$-a.s. It follows that~$\alpha$ takes $\pi$-a.s.\ values in $\gC \cap \gC_1$. Applying Lemma \ref{lemma:intersection-of-simplices} proves that $\alpha$ takes $\pi$-a.s.\ values in $\gC_A$.
\end{proof}

\begin{proof}[Proof of Theorem \ref{thm:Main1}]
By Lemma \ref{lemma:compatible-extension}, given an arbitrary element $x \in \Ext_{\iy}(\gC_A,\gC_B,\phi)$, there is a compatible sequence $(y_k)_{k \geq 0}$ with $y_1=x$. Equation \eqref{eq:main-compatible-extensions} applied with $k=1$ shows that $x=y_1 \in \gC_A \tmin \gC_B$.
\end{proof}

\begin{remark}
\normalfont
Theorem \ref{thm:Main1} can be extended to the case of more than two factors. Given proper cones $\gC_A,\gC_{B_1},\dots,\gC_{B_d}$ and for every $i \in [d]$ a linear form $\phi_i \in \inter(\gC_{B_i}^*)$, the minimal tensor product $ \gC_A \tmin \gC_{B_1} \tmin \dots \tmin \gC_{B_d}$ is equal to
\[\bigcap_{k_1,\dots k_d \geq 1}
\left( \Id_{V_A} \otimes \gamma_{k_1}^{\phi_1} \otimes \dots \otimes \gamma_{k_d}^{\phi_d} \right)\left( \gC_A \tmax \gC_{B_1}^{\tmax k_1} \tmax \dots \tmax \gC_{B_d}^{\tmax k_d} \right).
\]
The proof of this result is  similar to the proof of Theorem \ref{thm:Main1}. The corresponding classical de Finetti theorem is the following. Suppose that $(X_n^i)_{n \in \N, i \in [d]}$ is a family of random variables such that, for every finitely supported permutations $\sigma_1,\dots,\sigma_d$ of the integers, $(X_{\sigma_i(n)}^i)_{n \in \N, i \in [d]}$ has the same distribution as $(X_n^i)_{n \in \N, i \in [d]}$. Then~$(X^i_n)$ is distributed as a mixture of i.i.d.\ variables, each element in the mixture being distributed as an independent tuple $(Y_i)_{i \in [d]}$. We omit details.
\end{remark}

\section{Proof of Theorem~\ref{thm:Main2}}\label{sec:HierarchyFinite}

Fix a proper cone $\gC \subset V$ and $\phi \in \inter (\gC^*)$. The $k$th reduction map $\gamma_k^\phi : V^{\otimes k} \to V$ can be viewed as a tensor, namely the unique element of $(V^*)^{\otimes k} \otimes V$ such that
\begin{equation} \label{eq:gamma-pure-tensor} \langle \gamma_k^\phi, x_1 \otimes \dots \otimes x_k \otimes \psi \rangle = \frac{\psi(x_1) + \dots + \psi(x_k)}{k}.\end{equation}
for every  $x_1,\dots,x_k \in K_\phi$ and $\psi \in V^*$.

We say that the $k$th reduction map $\gamma_k^\phi$ is \emph{entanglement-breaking} if, as a tensor, it belongs to $(\gC^*)^{\tmin k} \tmin  \gC$. We use the following basic property of entanglement-breaking maps (see for instance \cite[Proposition 2.2]{AMH22}). Given proper cones $\gC_1 \subset V_1$ and $\gC_2 \subset V_2$ and a linear map $\Phi : V_1 \to V_2$, the following are equivalent:
\begin{enumerate}
\item[(i)] The map $\Phi$, as a tensor in $V_1^* \otimes V_2$, belongs to $\gC_1^* \tmin \gC_2$;
\item[(ii)] For every proper cone $\gC_A \subset V_A$, we have 
\[(\Id_{V_A} \otimes \Phi)(\gC_A \tmax \gC_1) \subset \gC_A \tmin \gC_2.\]
\end{enumerate}
Using this equivalence for the cones $\gC_1 = \gC^{\tmax k}$, $\gC_2 = \gC$ and the map $\Phi = \gamma_k^\phi$, we may reformulate Theorem \ref{thm:Main2} as follows.

\newtheorem*{theorem:restated}{Theorem~\ref{thm:Main2}'}

\begin{theorem:restated}
Let $\gC \subset V$ be a proper cone, $\phi \in \inter (\gC^*)$ and $k \geq 1$ an integer. Set $K_\phi = \gC \cap \phi^{-1}(1)$. The following are equivalent:
\begin{enumerate}
\item[(i)] The $k$th reduction map $\gamma_k^\phi$ is entanglement-breaking;
\item[(ii)] The base $K_\phi$ is affinely isomorphic to the Cartesian product of at most $k$ simplices.
\end{enumerate}
\end{theorem:restated}

If $P$ is a polytope, we denote by $\mathcal{V}(P)$ the set of its vertices and by $\mathcal{F}(P)$ the set of its facets (=faces of codimension $1$). We will start by proving the easy direction:

\begin{proof}[Proof that $(ii)$ implies $(i)$ in Theorem~\ref{thm:Main2}']
A basic observation is that, given two polytopes $P$ and $Q$, the faces of the Cartesian product $P \times Q$ have the form $F \times G$, where $F$ is face of $P$ and $G$ is face if $Q$. Consider simplices $\Delta_1,\dots,\Delta_k$, possibly $0$-dimensional, and the polytope $\Pi = \Delta_{1} \times \dots \times \Delta_{k}$. Its vertices have the form $(v_1,\dots,v_k)$ for $v_i \in \mathcal{V}(\Delta_{i})$ and its facets have the form
\begin{equation} \label{eq:facets-Pi}
\Delta_{1} \times \dots \times \Delta_{i-1} \times F \times \Delta_{i+1} \times \dots \times \Delta_{k} ,\end{equation}
for $i \in [k]$ and $F \in \mathcal{F}(\Delta_{i})$. Assume that $K_\phi$ is affinely isomorphic to $\Pi$. We may label the vertices of $K_\phi$ as 
\[ \{ x_{v_1,\dots,v_k} \st v_i \in \mathcal{V}(\Delta_i) \} \]
and define, for every $i \in [k]$ and $v \in \mathcal{V}(\Delta_i)$, an affine map $\psi_i^v : \aff(K_\phi) \to \R$ by
\[ \psi_i^v(x_{v_1,\dots,v_k}) = \begin{cases} 1 & \textnormal{ if } v_i=v \\ 0 & \textnormal{otherwise.} \end{cases} \]
We may extend $\psi_i^v$ to a linear map defined on $V$, which we still denote by $\psi_i^v$, and which is an element of $\gC^*$. We observe that for every $j \in [k]$,
\begin{equation} \label{eq:sum=phi} \phi = \sum_{v \in \mathcal{V}(\Delta_j)} \psi_j^v \end{equation}
since both sides evaluate to $1$ on any vertex of $K_\phi$.
Let $i \in [k]$, $v \in \mathcal{V}(\Delta_i)$ and $z_1,\dots,z_k \in K_\phi$. We compute, using the fact that $P_{\Sym_k(V)}^* = P_{\Sym_k(V^*)}$ and \eqref{eq:sum=phi} for every $j \neq i$, 
\begin{align*} 
 \mathllap{\left\langle 
\sum_{v_1 \in \mathcal{V}(\Delta_1)} \dots  \sum_{v_k \in \mathcal{V}(\Delta_k)}  P_{\Sym_k(V^*)} ( \psi_{1}^{v_1} \otimes \dots \otimes \psi_{k}^{v_k} )
\otimes x_{v_1,\dots,v_k}, z_1 \otimes \dots \otimes z_k \otimes \psi_i^v
\right\rangle}
\\
\begin{aligned}
 & = \sum_{v_1 \in \mathcal{V}(\Delta_1)} \dots  \sum_{v_k \in \mathcal{V}(\Delta_k)}
 {\bf 1}_{\{ v_i=v\}} \langle \psi_{1}^{v_1} \otimes \dots \otimes \psi_{k}^{v_k} , P_{\Sym_k(V)} (z_1 \otimes \dots \otimes z_k) \rangle 
\\ & =  \langle \phi^{\otimes (i-1)}  \otimes \psi_i^v \otimes \phi^{\otimes (k-i)}, P_{\Sym_k(V)} (z_1 \otimes \dots \otimes z_k ) \rangle 
\\ & =  \frac{\psi_i^v(z_1) + \dots + \psi_i^v(z_k)}{k}
\\ & = \langle \gamma_k^\phi, z_1 \otimes \dots \otimes z_k \otimes \psi_i^v \rangle.
\end{aligned}
\end{align*}
Since elements of the form $z_1 \otimes \dots \otimes z_k \otimes \psi_i^v$ span the space $V^{\otimes k} \otimes V^*$, we conclude that
\begin{equation*} 
\gamma_k^\phi = 
\sum_{v_1 \in \mathcal{V}(\Delta_1)} \dots  \sum_{v_k \in \mathcal{V}(\Delta_k)}
 P_{\Sym_k(V^*)} ( \psi_{1}^{v_1} \otimes \dots \otimes \psi_{k}^{v_k} )
\otimes x_{v_1,\dots,v_k}  \end{equation*}
and therefore the map $\gamma_k^\phi$ is entanglement-breaking.
\end{proof}

The second half of the proof of Theorem~\ref{thm:Main2}' requires a couple of preparatory lemmas and some terminology about polytopes. Let $P$ denote a polytope. The \emph{avoiding set} of a vertex $x \in \mathcal{V}(P)$ is given by
\[ 
\Av(x) = \{ F \in \mathcal{F}(P) \st x \not \in F \},
\] 
i.e., the set of facets not containing $x$. We say that a tuple $(F_1,\dots,F_k,x) \in  \mathcal{F}(P)^k \times \mathcal{V}(P)$ is \emph{admissible} if $\Av(x) \subset \{F_1,\dots,F_k\}$. 

\begin{lemma} \label{lemma:gamma-is-polytope}
If $\gamma_k^\phi$ is entanglement-breaking, then $K_\phi$ is a polytope.
\end{lemma}

\begin{proof}
By assumption, there exist an integer $N \geq 1$, elements $x_1,\dots,x_N$ in $K_\phi$ and $(f_i^j)_{i \in [N], j \in [k]}$ in~$\gC^*$ such that
\begin{equation} \label{eq:gamma-decomposition} \gamma_k^\phi = \sum_{i=1}^N f_i^1 \otimes \dots \otimes f_i^k \otimes x_i. \end{equation}
For every $y \in K_\phi$, we have 
\[ y = \gamma_k^\phi \left(y^{\otimes k}\right) = \sum_{i=1}^N \left( \prod_{j=1}^k f_i^j(y) \right) x_i .\]
It follows that $K_\phi \subset \conv \{ x_1,\dots,x_N \}$ and therefore $K_\phi=\conv \{ x_1,\dots,x_N \}$. In particular, $K_\phi$ is a polytope.
\end{proof}

We now assume that $\gamma^\phi_k$ is entanglement-breaking and hence that $K_\phi$ is a polytope. To every facet $F \in \mathcal{F}(K_\phi)$ we associate an element $\psi_F \in \gC^*$ with the property that $F = K_\phi \cap \ker \psi_F$ (this determines $\psi_F$ up to a positive scalar).

\begin{lemma} \label{lemma:gamma-vs-admissible}
If $\gamma_k^\phi$ is entanglement-breaking, then $\gamma_k^\phi$ belongs to the cone generated by elements of the form $\psi_{F_1} \otimes \cdots \otimes \psi_{F_k} \otimes x$, where $(F_1, \ldots, F_k,x) \in \mathcal{F}(K_\phi)^k \times \mathcal{V}(K_\phi)$ is admissible.
\end{lemma}
\begin{proof} 
For every facet $F \in \mathcal{F}(K_\phi)$, pick an arbitrary element $x_F$ in the relative interior of $F$. With this we define the \emph{vertex-facet tensor} as
\begin{equation} \label{equ:VFTensor}
\omega_k = \sum_{F \in \mathcal{F}(K_\phi)} x_F^{\otimes k} \otimes \psi_F \in V^{\otimes k}\otimes V^* .
\end{equation}
Since $\psi_F(x_F)=0$, it follows from \eqref{eq:gamma-pure-tensor} that $\langle \gamma_k^\phi , \omega_k \rangle = 0$. 

Given a tuple $(F_1,\dots,F_k,x) \in \mathcal{F}(K_\phi)^k \times \mathcal{V}(K_\phi)$, we have 
\[ \langle \psi_{F_1} \otimes \dots \otimes \psi_{F_k} \otimes x, \omega_k \rangle 
= \sum_{F \in \mathcal{F}(K_\phi)} \psi_F(x) \prod_{j=1}^k \psi_{F_j}(x_F) 
= \sum_{F \in \Av(x)} \psi_F(x) \prod_{j=1}^k \psi_{F_j}(x_F)
.\]
This quantity is nonnegative and vanishes if and only if, for every $F \in \Av(x)$, there is $j \in [k]$ such that $\psi_{F_j}(x_F)=0$. By definition of $x_F$, we have $\psi_{F_j}(x_F)=0$ if and only if $F=F_j$. We conclude that 
\begin{equation}\label{equ:vFtensPosProd} \langle \psi_{F_1} \otimes \dots \otimes \psi_{F_k} \otimes x, \omega_k \rangle \geq 0 \end{equation}
with equality if and only if $(F_1,\dots,F_k,x)$ is admissible. 

Since $\gC^* = \cone \{ \psi_F \st F \in \mathcal{F}(K_\phi) \}$, we may expand the decomposition \eqref{eq:gamma-decomposition} in the form 
\[
\gamma_k^\phi =  \sum_{F_1 \in \mathcal{F}(K_\phi)} \dots
\sum_{F_k \in \mathcal{F}(K_\phi)} \sum_{x \in \mathcal{V}(K_\phi)} \lambda_{F_1,\dots,F_k,x} \psi_{F_1} \otimes \dots \otimes \psi_{F_k} \otimes x
\]
for some $\lambda_{F_1,\dots,F_k,x} \geq 0$. Since $\langle \gamma_k^\phi,\omega_k \rangle =0$, we conclude from \eqref{equ:vFtensPosProd} that $\lambda_{F_1,\dots,F_k,x} = 0$ whenever $(F_1,\dots,F_k,x)$ is not admissible. This finishes the proof.
\end{proof}

\begin{proof}[Proof that $(i)$ implies $(ii)$ in Theorem~\ref{thm:Main2}']
Assume the map $\gamma_k^\phi$ to be entanglement-breaking. 
By Lemma \ref{lemma:gamma-vs-admissible} the base $K_\phi$ is a polytope. We now show that
\begin{equation}\label{equ:toShow}
\aff \left( \bigcap_{j=1}^n F_j \right) = \bigcap_{j=1}^n \aff ( F_j),
\end{equation}
for any integer $n \geq 1$ and $F_1,\dots,F_n \in \mathcal{F}(K_\phi)$. Since any face is an intersection of facets, \eqref{equ:toShow} implies that~$K_\phi$ satisfies the condition (ii) from Theorem \ref{theorem:polysimplices} and therefore is affinely equivalent to a product of simplices.


Note that $\aff \left( \bigcap F_j \right) \subset \bigcap \aff ( F_j)$. Conversely, consider $a \in  \bigcap \aff(F_j)$ and any $h \in V^*$ which vanishes on $\bigcap F_j$. For every $j \in [n]$, we have $\psi_{F_j}(a)=0$ since $a \in \aff(F_j)$. For any admissible $(G_1,\dots,G_k,y) \in  \mathcal{F}(K_\phi)^k \times \mathcal{V}(K_\phi)$, we have
\begin{equation} \label{eq:toShow2} \langle \psi_{G_1} \otimes \dots \otimes \psi_{G_k} \otimes y, a^{\otimes k} \otimes h \rangle = \psi_{G_1}(a) \dots \psi_{G_k}(a) h(y) = 0. \end{equation}
To prove \eqref{eq:toShow2}, observe that if $h(y) \neq 0$, then $y \not \in \bigcap F_j$ and hence there is an index $j \in [n]$ such that $y \not \in F_j$. This means that $F_j \in \Av(y)$ and therefore $F_j = G_i$ for some $i \in [k]$. We obtain that $\psi_{G_i}(a)=\psi_{F_j}(a)=0$, proving \eqref{eq:toShow2}.

By Lemma \ref{lemma:gamma-vs-admissible}, it follows that $\langle \gamma_k^\phi, a^{\otimes k} \otimes h \rangle = 0$. By \eqref{eq:gamma-pure-tensor}, we have $h(a) = \langle \gamma_k^\phi, a^{\otimes k} \otimes h \rangle$ and therefore $h(a)=0$. Since every $h \in V^*$ which vanishes on $\bigcap F_j$ vanishes at $a$, we have $a \in \mathspan(\bigcap F_j)$. We conclude the proof of \eqref{equ:toShow} by observing that $a \in \aff(K_\phi) \cap \mathspan(\bigcap F_j) = \aff(\bigcap F_j)$.

Theorem~\ref{theorem:polysimplices} tells us that $K_\phi$ is affinely isomorphic to $\Pi = \Delta_{1} \times \dots \times \Delta_{l}$ for some integer $l \geq 1$ and nontrivial simplices $\Delta_1,\dots,\Delta_l$; we still need to show that $l \leq k$.
By Lemma~\ref{lemma:gamma-vs-admissible}, there is at least one admissible tuple $(F_1,\dots,F_k,x) \in \mathcal{F}(K_\phi)^k \times \mathcal{V}(K_\phi)$. In particular,  $\card \Av(x) \leq k$. On the other hand, we observe from \eqref{eq:facets-Pi} that the polytope $\Pi$ has the property that the avoiding set of any vertex contains exactly $l$ elements. It follows that $l \leq k$.
\end{proof}

\subsection*{Further remarks:}
We finish this section with two remarks about concepts appearing in the previous proof.

\smallskip

(a) It is easy to see that if $K$ is the product of $k$ nontrivial simplices, then $\card (\Av\lb x\rb) = k$ for every vertex $x\in \mathcal{V}(K)$. Conversely, it would be interesting to determine which polytopes $K$ have the property that $\card (\Av\lb x\rb) = k$ for every vertex $x\in \mathcal{V}(K)$. For $k=1$ there are only simplices. It was observed by Martin Winter that this class of polytopes is closed under taking free joins~\cite{winterMO}, and that for $k=2$ this property characterizes iterated free joins of Cartesian products of two nontrivial simplices~\cite{WinterPC}.

\smallskip

(b) The vertex-facet tensor $\omega_k$ from
\eqref{equ:VFTensor} is an interesting object encoding certain combinatorial properties of the polytope $K_\phi$. For example, we have the following proposition:

\begin{proposition}\label{cor:VFtensorInclus}
Consider a polyhedral cone $\gC$ with base $K_\phi=\phi^{-1}(1)\cap \gC$ for some functional $\phi\in \inter(\gC^*)$ and an integer $k \geq 1$. Then, the following are equivalent:
\begin{enumerate}
\item[(i)] We have
\[
\omega_k\in \inter\lb \gC^{\otimes_{\max} k}\otimes_{\max} \gC^*\rb;
\]
\item[(ii)] For every vertex $x \in \mathcal{V}(K_\phi)$, we have $\card (\Av\lb x\rb)>k$.
\end{enumerate}
\end{proposition}

\begin{proof}
For any proper cone $\gC_0 \subset V_0$, the interior of $\gC_0^*$ is the set of linear forms $f \in V_0^*$ such that $f(x)>0$ for every $x \in \gC_0\setminus \{0\}$. Using this observation for the cone $\gC_0 = \left(\gC^*\right)^{\tmin k} \tmin \gC$ shows that (i) is equivalent to the statement
\[ \forall F_1,\dots,F_k \in \mathcal{F}(K_\phi),\ \forall x \in \mathcal{V}(K_\phi),\ \exists F \in \mathcal{F}(K_\phi) \setminus \{F_1,\dots,F_k\} \st \psi_F(x) > 0,\]
which is equivalent to (ii).
\end{proof}

Since $\langle \gamma_k^\phi , \omega_k \rangle = 0$ always holds, Proposition \ref{cor:VFtensorInclus} implies directly that $\gamma_k^\phi$ is not entanglement-breaking whenever $\card (\Av\lb x\rb) > k$ for every vertex $x\in \mathcal{V}(K_\phi)$.

\section{Proof of Theorem \ref{theorem:polysimplices}} \label{sec:PolySimpl}

Before proving Theorem \ref{theorem:polysimplices}, we recall some facts about polytopes. Given a polytope $P$ and a vertex $v \in \mathcal{V}(P)$, we denote by $P/v$ the \emph{vertex figure} of $P$ at $v$ (see \cite[Chapter 2.1]{Grunbaum} for definition). A $d$-polytope~$P$ is said to be \emph{simple} if any vertex is contained in exactly $d$ facets, or equivalently if all its vertex figures are $(d-1)$-simplices. A polytope $P$ is said to be $2$-\emph{level} if for every facet $F \in \mathcal{F}(P)$, all vertices not in $F$ lie in the same translate of $\aff(F)$. Since the complement of any facet in the graph of $P$ is connected \cite[p.\ 475]{Sallee67}, $P$ is $2$-level if and only if any face disjoint from any facet $F \in \mathcal{F}(P)$ lies in a translate of $\aff(F)$.

We use the following characterization of products of simplices. For characterizations of products of simplices up to combinatorial equivalence, see \cite{YuMasuda21}.

\begin{theorem}[Kaibel--Wolff, \cite{KaibelWolff00}] \label{theorem:kaibelwolff}
A polytope is affinely equivalent to a product of simplices if and only if it is simple and $2$-level.
\end{theorem}

A polytope $P \subset \R^n$ is said to be a \emph{0/1-polytope} if its vertices are a subset of $\{0,1\}^n$. Kaibel and Wolff proved in \cite{KaibelWolff00} that a simple 0/1-polytope is a product of simplices. Since any $2$-level polytope is affinely equivalent to a 0/1-polytope (see~\cite{GPT10}), this implies Theorem \ref{theorem:kaibelwolff}.

\begin{proof}[Proof of Theorem \ref{theorem:polysimplices}]
The fact that a simplex satisfies (ii) is a consequence of the following observation: if $(x_i)_{i \in I}$ are affinely independent, then for every $J_1$, $J_2 \subset I$, we have
\[ \aff(\{x_i \st i \in J_1 \cap J_2\}) = \aff(\{x_i \st i \in J_1\}) \cap \aff(\{x_i \st i \in J_2\}). \]
By induction, to prove (i) $\Longrightarrow$ (ii), it suffices to prove that the class of polytopes satisfying (ii) is closed under products. Let $P$ and $Q$ be polytopes satisfying (ii) and consider a family $(F_i)_{i \in I}$ of faces of $P \times Q$. There are faces $(G_i)_{i \in I}$ of $P$ and $(H_i)_{i \in I}$ of $Q$ such that $F_i=G_i\times H_i$ for every $i \in I$. Using the relation $\aff(A \times B) = \aff(A) \times \aff(B)$ whenever $A$, $B$ are subsets of vector spaces, we obtain
\begin{multline*}
\aff \left(\bigcap_{i \in I} F_i \right)  =  \aff \left( \bigcap_{i\in I} G_i \times \bigcap_{i\in I} H_i \right)   =  \aff \left( \bigcap_{i\in I} G_i \right) \times \left( \bigcap_{i\in I} H_i \right) \\
 = \left(\bigcap_{i \in I} \aff(G_i) \right) \times \left( \bigcap_{i \in I} \aff(H_i) \right)  = \bigcap_{i \in I} \aff(G_i) \times \aff(H_i) = \bigcap_{i \in I} \aff(F_i)
\end{multline*}
and the implication (i) $\Longrightarrow$ (ii) follows.

Conversely, let $P$ be an $n$-polytope satisfying (ii). We first show that $P$ is simple. Consider a vertex $x \in \mathcal{V}(P)$ and facets $F_1,\dots,F_{n-1} \in \mathcal{F}(P)$ containing~$x$. The set $\bigcap \aff(F_i)$ is the nonempty intersection of $n-1$ affine hyperplanes, hence has dimension $\geq 1$. By \eqref{eq:aff}, the face $\bigcap F_i$ has dimension $\geq 1$. By the 1-1 correspondence between faces of $P$ containing $x$ and faces of $P/x$ (see \cite[Proposition 2.4]{Grunbaum}), the $(n-1)$-polytope $P/x$ is dual-$(n-1)$-neighbourly, i.e., any $n-1$ facets have a common point. Since a $d$-polytope which is dual-$k$-neighbourly for $k> \lfloor d/2 \rfloor$ is a simplex (see \cite[p.~123]{Grunbaum}), it follows that $P/x$ is a simplex for every $x \in \mathcal{V}(P)$ and therefore $P$ is simple.

Let $F \in \mathcal{F}(P)$ and $G$ be a face of $P$ such that $F \cap G = \emptyset$. By \eqref{eq:aff}, we have $\aff(F) \cap \aff(G) = \emptyset$ and therefore $G$ lies in a translate of $\aff(F)$. This means that $P$ is $2$-level. By Theorem \ref{theorem:kaibelwolff}, it follows that $P$ is affinely equivalent to a product of simplices.
\end{proof}


\subsection*{Acknowledgements} GA was supported in part by ANR (France) under the grant ESQuisses (ANR-20-CE47-0014-01). AMH acknowledges funding from the European Union's Horizon 2020 research and innovation programme under the Marie Sk\l odowska-Curie Action TIPTOP (grant no.\ 843414). MP was supported by the Deutsche Forschungsgemeinschaft (DFG, German Research Foundation, project numbers 447948357 and 440958198), the Sino-German Center for Research Promotion (Project M-0294), the ERC (Consolidator Grant 683107/TempoQ) and the Alexander von Humboldt Foundation. We thank Ludovico Lami and Christophe Sabot for useful discussions.

\subsection*{Conflicts of interest} The authors declare no conflict of interest.

\appendix \label{app:quantum}

\section{Proof of Proposition \ref{prop:quantum}}

We denote by $(\ket{1},\ket{2},\ket{3})$ the canonical basis of $\CC^3$. As common in the quantum information theory literature, we write $\ket{ij}$ as a shortcut for $\ket{i} \otimes \ket{j}$. Consider the operators 
\begin{eqnarray*} X_{\alpha,\beta,\gamma} 
&= &\alpha \sum_i \ketbra{ii}{ii} + \beta \sum_{i\neq j} \ketbra{ij}{ij} + \gamma \sum_{i \neq j} \ketbra{ii}{jj} \\
& = & { \footnotesize\begin{pmatrix} 
\alpha & & & & \gamma & & & & \gamma \\
& \beta & & & & & & & \\
& & \beta & & & & & & \\
& & & \beta & & & & & \\
\gamma & & & & \alpha & & & & \gamma \\
& & & & & \beta & & & \\
& & & & & & \beta & & \\
& & & & & & & \beta & \\
\gamma & & & & \gamma & & & & \alpha 
\end{pmatrix}.}
\end{eqnarray*}
with parameters $\alpha$, $\beta$, $\gamma\in \R$.
This operator is positive semidefinite if $\beta \geq 0$ and $\alpha \geq \gamma \geq 0$. Set $\eta = 1-\sqrt{2}/2$. Consider the operator
\[ Y = X_{1,\eta,1} = \frac{1}{4} X_{4,1,2\sqrt{2}+1} + \frac{3-2\sqrt{2}}{4} X_{0,1,1}\in (\M_3\otimes \M_3)^+
\]
Observe the following facts.
\begin{itemize}
    \item The operator $X_{4,1,2\sqrt{2}+1}$ is $2$-PSD-extendible, an extension being given by $\ketbra{\psi_1}{\psi_1}+\ketbra{\psi_2}{\psi_2}+\ketbra{\psi_3}{\psi_3}$, with
    \begin{align*} \ket{\psi_1} = \sqrt{2} \ket{111} + \ket{212} + \ket{221} + \ket{313} + \ket{331} , \\
    \ket{\psi_2} = \sqrt{2} \ket{222} + \ket{323} + \ket{332} + \ket{121} + \ket{112} ,\\
     \ket{\psi_3} = \sqrt{2} \ket{333} + \ket{131} + \ket{113} + \ket{232} + \ket{223} .
     \end{align*}
    \item The operator $X_{0,1,1}$ is $2$-max-extendible, since its partial transpose is $2$-PSD-extendible, an extension being given by $\ketbra{\psi}{\psi}$, with
    \[ \ket{\psi} =  \ket{123} + \ket{132} + \ket{213} + \ket{231} + \ket{312} + \ket{321}.\]
    \item By the two previous points, the operator $Y$ is $2$-max-extendible as a positive combination of $2$-max-extendible operators.
    \item To show that $Y$ is not 2-PSD-extendible, consider the operator $W = X_{1,\eta,-2\eta}$. It is easy to check numerically that the operator 
    \[ W_2 = \lb\Id_{\M_3} \otimes P_{\Sym_2(\M_3)} \rb (W \otimes {\bf 1}_3) \]
    is positive definite. Assume that $Y$ is $2$-PSD-extendible, with extension $Y_2 \in \lb\M_3^{\otimes 3}\rb^+$. We would then have
    \[ \tr(Y_2W_2)=\tr(YW) = 3 \cdot 1 + 6 \cdot \eta^2 + 6 \cdot (-2\eta) = 0, \]
    showing that $Y_2=0$, leading to a contradiction.
\end{itemize}

\bibliography{monogamy-GPT}{}
\bibliographystyle{plain}

\end{document}